\documentclass{llncs}
\usepackage[ruled, vlined]{algorithm2e}

\usepackage{enumerate}
\usepackage{amssymb}
\usepackage{amsmath}
\usepackage{stmaryrd}
\usepackage{graphics}
\usepackage{graphicx}
\usepackage{capt-of}
\usepackage{subfig}
\usepackage{colortbl}
\usepackage{xcolor}
\usepackage{wrapfig}
\usepackage{capt-of}
\usepackage{colortbl,hhline}
\newcommand{\bic}[1]
{
\cellcolor{lightgray}{\bf #1}
}
\usepackage{fancybox}

\def\tu#1{{\langle #1\rangle}}

\def\noticka#1{}

  \def\X{$\times$}
  \def\O{}

\title{Mining Biclusters of Similar Values\\with Triadic Concept Analysis}

\author{Mehdi Kaytoue\inst{1}, Sergei O. Kuznetsov\inst{2}, Juraj Macko\inst{3},\\ Wagner Meira Jr.\inst{1} and  Amedeo Napoli\inst{4}}
\authorrunning{Kaytoue et al.} 
\institute{
Universidade Fereral de Minas Gerais -- Belo Horizonte -- Brazil\\
\and HSE -- Pokrovskiy Bd. 11 -- 109028 Moscow -- Russia 
\and Palacky University -- 17. listopadu -- 77146 Olomouc -- Czech Republic\\
\and INRIA/LORIA -- Campus Scientifique, B.P. 239 -- Vand{\oe}uvre-l\`{e}s-Nancy -- France
\email{kaytoue@dcc.ufmg.br, kuznetsovs@yandex.ru, juraj.macko@upol.cz, meira@dcc.ufmg.br, napoli@loria.fr}
}
\begin{document}
\maketitle
\begin{abstract}
Biclustering numerical data became a popular data-mining task in the beginning of 2000's, especially for analysing gene expression data. A bicluster reflects a strong association between a subset of objects and a subset of attributes in a numerical object/attribute data-table. So called biclusters of similar values can be thought as maximal sub-tables with close values. Only few methods address a complete, correct and non redundant enumeration of such patterns, which is a well-known intractable problem, while no formal framework exists. In this paper, we introduce important links between biclustering and formal concept analysis. More specifically, we originally show that Triadic Concept Analysis (TCA), provides a nice mathematical framework for biclustering. Interestingly, existing algorithms of TCA, that usually apply on binary data, can be used (directly or with slight modifications) after a preprocessing step for extracting maximal biclusters of similar values.
\end{abstract}

\section{Introduction}
\label{introduction}

Numerical data biclustering mainly appeared in the beginning of 2000's as a first answer to new challenges raised by biological data analysis, and especially gene expression data analysis~\cite{MadeiraO04}. Starting from an object/attribute numerical data-table (e.g. Table~\ref{tab:data}), the goal is to group together some objects with some attributes according to the values taken by these attributes for these objects~\cite{MadeiraO04}. Accordingly, a bicluster is formally defined as a pair composed of a set of objects and a set of attributes. Such pair can be represented as a rectangle in the numerical table, modulo lines and columns permutations. Table~\ref{tab:data} is a numerical dataset with objects in lines and attributes in columns, while each table entry corresponds to the value taken by the attribute in column for the object in line. Table~\ref{tab:simbic} illustrates bicluster $(\{g_1,g_2,g_3\},\{m_1,m_2,m_3\})$  as a grey rectangle.

There are several types of biclusters in the literature (see~\cite{MadeiraO04} for a  survey), depending on the relation between the values taken by their attributes for their objects. The most simple case can be understood as rectangles of equal values: a bicluster corresponds to a set of objects whose values taken by a same set of attributes are exactly the same, e.g. $(\{g_1,g_2,g_3\},\{m_5\})$. Constant biclusters  only appear in idyllic situations: generally numerical data are noisy. Accordingly, a straightforward generalization of such biclusters lies in so called biclusters of similar values: they are represented by rectangles with almost identical, say similar, values~\cite{MadeiraO04,BessonRRB06,KaytoueKN11}. Table~\ref{tab:simbic} illustrates a bicluster  of similar values $(\{g_1,g_2,g_3\},\{m_1,m_2,m_3\})$ where two values are said to be similar if their difference is no more than $1$. Moreover, this bicluster is maximal: neither an object nor an attribute can be added without violating the similarity condition.

Only few methods address a complete, correct and non redundant enumeration of such patterns~\cite{BessonRRB06,KaytoueKN11}, which is a well-known intractable problem~\cite{MadeiraO04}, while no formal framework exists. In this paper, we show that Formal Concept Analysis (FCA)~\cite{GanterW99}, and especially Triadic Concept Analysis (TCA)~\cite{LehmannW95} provides a suitable and well defined framework for this task: Basically, an object has an attribute under a condition (a value). After a simple scaling procedure (turning original data into binary), a bicluster is represented as a triadic concept, composed of a set of objects, a set of attributes (both characterizing the corresponding ``rectangle'') and a set of values. All sets are maximal thanks to existing concept forming derivation operators of TCA. This comes with several advantages:
\vspace{-0.1cm}
\begin{itemize}
\item Two values $w_1,w_2$ of the original data are said to be similar iff their difference does not exceed a given parameter $\theta$. In this case, we write $w_1 \simeq_\theta  w_2 \iff \vert w_1 - w_2 \vert \leq \theta$. Otherwise, we write $w_1 \not \simeq_\theta  w_2$. The trilattice produced with TCA after scaling gives all maximal biclusters of similar values for any $\theta$ ordered w.r.t. similarity of their values.

\item The well known notion of \textit{frequency} takes a semantics w.r.t. similarity of values. For example, let $(A,B,C)$ be a triconcept, where $A$ is a set of objects, $B$ a set of attributes, and $C$ a set of similar values. Assume $(A,B)$ to be the corresponding bicluster. The higher $|C|$, the more similar are the values of the bicluster. If all $|A|$, $|B|$, and $|C|$ are high we obtain a bicluster represented as a large rectangle of close values.

\item Existing algorithms from TCA~\cite{JaschkeHSGS06} and $n$-ary closed set mining~\cite{CerfBRB09} can be used directly after scaling. We also provide a new algorithm to compute biclusters maximal only for a given  $\theta$ (see algorithm {\sc TriMax} later on).


\item Both scaling procedure and algorithm {\sc TriMax} computations can be directly distributed to several computing cores.

\item The method can be  adapted to $n$-ary numerical datasets. For example, with $n=3$, a $n$-cluster would be a maximal $3D$-box of similar values. It can be applied to 3D gene expression data, monitoring the behaviour of genes in different samples over time. It follows that
mining $n$-dimensional clusters can be achieved with $n+1$-adic concept analysis.
\end{itemize}

The paper is organized as follows. Firstly, preliminaries regarding TCA are presented in Section~\ref{tca}. Then Section~\ref{problem} formally states the problem. It is followed by the description of our two methods, respectively in Section~\ref{method1} and~\ref{method2}. The first shows how TCA can help characterizing all maximal biclusters for any $\theta$, while the second restricts the problem to a user-given $\theta$. This is followed by experiments on the proposed approaches. Finally, the paper ends with a discussion and perspectives of further research.

\vspace{-0.7cm}
\begin{table}
\begin{minipage}{0.5\textwidth}\centering
\caption{A numerical dataset}
\begin{scriptsize}
\begin{tabular}{c||ccccc}
  & $m_1$ & $m_2$ & $m_3$ & $m_4$ & $m_5$ \\
\hline \hline $g_1$ & 1 & 2 & 2 & 1 & 6 \\
$g_2$ & 2 & 1 & 1 & 0 & 6 \\
$g_3$ & 2 & 2 & 1 & 7 & 6 \\
$g_4$ & 8 & 9 & 2 & 6 & 7 \\
\end{tabular}
\end{scriptsize}
\label{tab:data}
\end{minipage}
\hfill%
\begin{minipage}{0.49\textwidth}\centering
\caption{A bicluster of similar values}
\begin{scriptsize}
\begin{tabular}{c||ccccc}
  & $m_1$ & $m_2$ & $m_3$ & $m_4$ & $m_5$ \\
\hline \hline $g_1$ & \bic{1} & \bic{2} & \bic{2} & 1 & 6 \\
$g_2$ & \bic{2} & \bic{1} & \bic{1} & 0 & 6 \\
$g_3$ & \bic{2} & \bic{2} & \bic{1} & 7 & 6 \\
$g_4$ & 8 & 9 & 2 & 6 & 7 \\
\end{tabular}
\end{scriptsize}
\label{tab:simbic}
\end{minipage}%
\end{table}

\vspace{-0.8cm}


\section{Triadic Concept Analysis}
\label{tca}
We assume that the reader is familiar with basic notions of Formal Concept Analysis~\cite{GanterW99}.
Lehmann and Wille introduced Triadic Concept Analysis (TCA~\cite{LehmannW95}). Data are represented by a  triadic context, given by $(G,M,B,Y)$. $G$, $M$, and $B$ are respectively called sets of objects, attributes and conditions, and  $Y \subseteq G \times M \times B$. The fact $(g,m,b) \in Y$ is interpreted as the statement ``Object $g$ has the attribute $m$ under condition $b$''.

A (triadic) concept of $(G,M,B,Y)$ is a triple $(A_1,A_2,A_3)$ with $A_1 \subseteq G$, $A_2 \subseteq M$ and $A_3 \subseteq B$ satisfying the two following statements: (i) $A_1 \times A_2  \times A_3 \subseteq Y$, $X_1 \times X_2  \times X_3 \subseteq Y$ and (ii) $A_1 \subseteq X_1$, $A_2 \subseteq X_2$ and $A_3 \subseteq X_3$ implies $A_1=X_1$, $A_2=X_2$ and $A_3=X_3$. If $(G,M,B,Y)$ is represented by a three dimensional table, (i) means that a concept stands for a  3-dimensional rectangle full of crosses while $(ii)$ characterises component-wise maximality of concepts. For a triadic concept $(A_1,A_2,A_3)$, $A_1$ is called the extent, $A_2$ the intent and $A_3$ the modus.

To describe the derivation operators, it is convenient to alternatively represent a triadic context as $(K_1,K_2,K_3,Y)$. Then, for $\{i,j,k\}=\{1,2,3\}$, $j < k$, $X \subseteq K_i$ and $Z \subseteq K_j \times K_k$, $(i)$-derivation operators are defined by: \\
\mbox{~~~~~~~~~~}$\Phi: X \rightarrow X^{(i)} :  \{ (a_j,a_k) \in K_j \times K_k~|~(a_i,a_j,a_k) \in Y  \mbox{ for all } a_i \in X \} $  \\
\mbox{~~~~~~~~~~}$\Phi^{'}: Z \rightarrow Z^{(i)}:  \{  a_i \in K_i ~| ~(a_i,a_j,a_k) \in Y  \mbox{ for all } (a_j,a_k) \in Z \} $ 

This definition leads to derivation operator $ \mathbf{K}^{(3)} $ and dyadic context $ \mathbf{K}^{(3)} =\tu{K_3, K_1 \times K_2, Y^{(3)}} $. Further derivation operators are defined as follows: for  $\{i,j,k\}=\{1,2,3\}$, $X_i \subseteq K_i$, $X_j \subseteq K_j$ and $A_k \subseteq K_k$, the $(i,j,A_k)$-derivation operators are defined by:\\ 
\mbox{~~~~~~}$\Psi: X_i \rightarrow X^{(i,j,A_k)}_i :  \{ a_j \in K_j~ |  ~ (a_i,a_j,a_k) \in Y  \mbox{ for all } (a_i,a_k) \in X_i \times A_k \} $  \\
\mbox{~~~~~~}$\Psi^{'}: X_j \rightarrow X^{(i,j,A_k)}_j :  \{ a_i \in K_i ~|~(a_i,a_j,a_k) \in Y  \mbox{ for all } (a_j,a_k) \in X_j \times A_k \} $

Operators $\Phi$ and $\Phi^{'}$ will be called outer operators, pair of both operators outer closure and dyadic operators $\Psi$ and $\Psi^{'}$ inner operators or inner closure when pair of both is used.  Derivation operators of dyadic context are defined by $ \mathbf{K}^{ij}_{A_k} =\tu{K_i,K_j,Y^{ij}_{A_k}} $, where $(a_i,a_j)\in Y^{ij}_{A_k}$ iff $ a_i,a_j,a_k \mbox{ are related by } Y  \mbox{ for all } a_k \in A_k $.
 
From a computational point of view, \cite{JaschkeHSGS06} developed  the algorithm {\sc Trias} for extracting frequent triadic concepts, i.e. whose extent, intent and modus cardinalities are higher than user-defined thresholds (see also \cite{JiTT06}). Cerf et al. presented a more efficient algorithm called {\sc Data-peeler} able to handle $n$-ary relations~\cite{CerfBRB09} while formal definitions lie in so called Polyadic Concept Analysis~\cite{Voutsadakis02}.

\section{Notations and problem settings}
\label{problem}

A numerical dataset is realized by a many-valued context~\cite{GanterW99} and we define accordingly (maximal) biclusters of similar values.

\begin{definition}[Many-valued context] Let $G$ be a set of objects, $M$ be a set of attributes, $W$ be the set of attribute values and $I$ be a ternary relation defined on the Cartesian product $G \times M \times W$. The fact $(g,m,w) \in I$, also written $m(g)=w$, means that ``Attribute $m$ takes the value $w$ for the object $g$''. The tuple $(G,M,W,I)$ is called many-valued context, or simply numerical dataset in this paper.\end{definition} 
\begin{example}
Table~\ref{tab:data} is a numerical dataset, or many-valued context, with objects $G = \{g_1,g_2,g_3,g_4\}$, attributes $M = \{m_1,m_2,m_3,m_4,m_5\}$, $W = \{0,1,2,6,7,8,9\}$ and for example $m_5(g_2) = 6$.
\end{example}

\begin{definition}[Bicluster] 
In a numerical dataset $(G,M,W,I)$, a bicluster is a tuple $(A,B)$ with $A  \subseteq G$ and $B \subseteq M$.
\end{definition} 

\begin{definition}[Similarity relation and bicluster of similar values] 
Let $w_1, w_2 \in W$ be two attribute values and $\theta \in \mathbb{N}$ be a user-defined parameter, called \textit{similarity parameter}. $w_1$ and $w_2$ are said to be similar iff $\vert w_1 - w_2 \vert \leq \theta$ and we note $w_1 \simeq_\theta w_2$.
$(A,B)$  is bicluster of similar values if $m(g) \simeq_{\theta} n(h)$ for all $g,h \in A$ and for all  $m,n \in B$. 
\end{definition} 

\begin{definition}[Maximal bicluster of similar values] 
A bicluster of similar values  $(A,B)$  is maximal if adding either an object in $A$ or an attribute in $B$ does not result in a bicluster of similar values.
\end{definition}

 \begin{example}[From Table~\ref{tab:data}]
$(\{g_1,g_4\}, \{m_2, m_4\})$ is a bicluster. $(\{g_1,g_2\}, \{m_2\})$ is a bicluster of similar values with $\theta \geq 1$. However, it is not maximal. With \mbox{$ 1 \leq \theta < 5$},  $(\{g_1,g_2,g_3\}, \{m_1,m_2,m_3\})$ is maximal. Finally, with $\theta = 7$ the bicluster $(\{g_1,g_2,g_3\}, \{m_1,m_2,m_3,m_4,m_5\})$ is maximal. Note that a constant (maximal) bicluster is a (maximal) bicluster of similar values with $\theta =0$.
\end{example}

Thus the problem that we address in this paper is the extraction of all maximal biclusters of similar values from a numerical dataset. We desire the extraction to be complete, correct and non-redundant compared to several existing methods of the literature based on heuristics~\cite{MadeiraO04}.  For that matter, we propose in the next section a first method aiming at extracting biclusters for any similarity parameter $\theta$. This method establishes new links between biclustering and FCA in general, and TCA in particular. Then, the present methodology is adapted to characterize and extract biclusters that are maximal for a given $\theta$ only as usually done in the literature~\cite{BessonRRB06,KaytoueKN11,MadeiraO04}.

\section{Biclusters of similar values in Triadic Concept Analysis}
\label{method1}

Firstly, we consider the problem of generating maximal biclusters for any $\theta$. Starting from a numerical dataset $(G,M,W,I)$, the basic idea lies in building a triadic context $(G,M,T,Y)$ where the two first dimensions remain formal objects and formal attributes, while $W$ is scaled into a third dimension denoted by $T$.  This new dimension $T$ is called the \textit{scale dimension}: intuitively, it gives different ``spaces of values'' that each object-attribute pair $(g,m) \in G \times M$ can take. Once the scale is given, a triadic context is derived from which triadic concepts are characterized.

We use the \textit{interordinal scaling}~\cite{GanterW99} to build the scale dimension. It allows to encode in $2^T$ all possible intervals of values in $W$. This scale allows to derive a triadic context from which any bicluster of similar values can be characterized as a triadic concept. We made more precise these statements and illustrate the whole procedure with examples.

\begin{definition}[Interordinal Scaling]
A scale is a binary relation $J \subseteq W \times T$ associating original elements from the set of values $W$ to their derived elements in $T$. In the case of interordinal scaling, $T = \{[min(W),w], \forall w \in W\} \cup \{[w,max(W)], \forall w \in W\} $. Then $(w,t) \in J $ iff $w \in t$.
\end{definition} 
 
\begin{example}
Table~\ref{tab:scale2} gives the tabular representation of the interordinal scale for Table~\ref{tab:data}. Intuitively, each line describes a single value, while dyadic concepts represent all possible intervals over $W$. An example of dyadic concept in this table is given by $(\{6,7,8\},\{t_6,t_7,t_8,t_9,t_{10}\})$, rewritten as $(\{6,7,8\},\{[6,8]\})$ since $\{t_6,t_7,t_8,t_9,t_{10}\}$ represents the interval $[0,8] \cap [0,9] \cap [1,9] \cap [2,9] \cap [6,9] = [6,8]$.
\vspace{-0.5cm}
\begin{table}
  \def\X{$\times$}
  \def\O{}
  \centering
  \renewcommand{\arraystretch}{1.1}
  \def\rr#1{\rotatebox{90}{#1}}

\begin{scriptsize}
\begin{tabular}{|r|ccccccccccccc|}
\hline
$J$ & \rotatebox{90}{$t_1 = [0,0]$~} & \rotatebox{90}{$t_2= [0,1]$~} & \rotatebox{90}{$t_3= [0,2]$~}& \rotatebox{90}{$t_4= [0,6]$~} & \rotatebox{90}{$t_5= [0,7]$~} & \rotatebox{90}{$t_6= [0,8]$~} & \rotatebox{90}{$t_{7}= [0,9]$~} & \rotatebox{90}{$t_{8}= [1,9]$~} & \rotatebox{90}{$t_{9}= [2,9]$~} & \rotatebox{90}{$t_{10}= [6,9]$~} & \rotatebox{90}{$t_{11}= [7,9]$~} & \rotatebox{90}{$t_{12}= [8,9]$~}  &\rotatebox{90}{$t_{13}= [9,9]$~} \\
\hline
\hline
\multicolumn{1}{|c|}{$0$} & \X     & \X     & \X       & \X     & \X     & \X     & \X     & \O     & \O      & \O     & \O     & \O     & \O \\

\multicolumn{1}{|c|}{$1$} & \O     & \X     & \X     & \X        & \X     & \X     & \X     & \X     & \O     & \O     & \O     & \O     & \O    \\

\multicolumn{1}{|c|}{$2$} & \O     & \O     & \X     & \X     & \X        & \X     & \X     & \X     & \X     & \O     & \O     & \O     & \O    \\
\multicolumn{1}{|c|}{$6$} & \O     & \O     & \O     & \X     & \X     & \X     & \X     & \X     & \X      & \X     & \O     & \O     & \O \\

\multicolumn{1}{|c|}{$7$} & \O     & \O     & \O     & \O     & \X     & \X & \X     & \X     & \X         & \X     & \X     & \O     & \O \\

\multicolumn{1}{|c|}{$8$} & \O     & \O     & \O     & \O     & \O     & \X     & \X     & \X     & \X      & \X     & \X     & \X     & \O \\

\multicolumn{1}{|c|}{$9$} & \O     & \O     & \O     & \O     & \O     & \O     & \X     & \X     & \X       & \X     & \X     & \X     & \X \\
\hline
\end{tabular}%
\end{scriptsize}

\vspace{0.3cm}
   \caption{Interordinal scale of the set of attribute values $W$.}
   \label{tab:scale2}
\end{table}
\vspace{-1.1cm}
\end{example}
 
Once the scale is defined, we can derive the triadic context w.r.t. this scale. 

\begin{definition}[Triadic scaled context] 
Let $Y$ be ternary relation $Y \subseteq G \times M \times T$. Then $(g,m,t) \in Y$ iff  $(m(g) , t) \in J$, or simply $m(g) \in t$.
We call the tuple $(G,M,T,Y)$ the triadic scaled context of the numerical dataset $(G,M,W,I)$.
\end{definition} 

\begin{example}
The object-attribute pair $(g_1,m_1)$ taking value  $m_1(g_1) = 1$ is scaled into triples $(g_1,m_1,t) \in Y$ where $t$ takes any interval in  $\{[0,1],[0,2], [0,6], [0,7]$, $ [0,8],[0,9], [1,9]\}$. The intersection of intervals in this set is the original value itself, i.e. $m_1(g_1) = 1$, a basic property of interordinal scaling. As a result, Table~\ref{tab:fc3d} illustrates the whole scaled triadic context derived from the numerical dataset given in Table~\ref{tab:data} using interordinal scale. The very first cross ($\times$) in this table (upper left) represents the tuple $(g_2,m_4,t_1)$, meaning that $m_4(g_2) \in [0,0]$.
\end{example}

We present now our first main result: there is a one-to-one correspondence between (i) the set of maximal biclusters of similar values in a given numerical dataset for any similarity parameter $\theta$ and (ii) the set of all triadic concepts in the triadic context derived with interordinal scaling. 

\begin{proposition}
Tuple $ \tu{A,B,U} $, where $A \subseteq G$, $B \subseteq G$ and $U \subseteq T$ is triadic concept iff $(A,B)$ is a maximal bicluster of similar values for some $\theta \geq 0$.
\end{proposition} 
\begin{proof}
We leave the proof in the Appendix of the paper since we need to introduce notations and propositions not necessary in the rest of the paper.
\end{proof}

\begin{example}
For example, $( \{ g_1, g_2,g_3  \}  , \{ m_1, m_2, m_3  \}    ,  \{ t_3, t_4, t_5, t_6, t_7 ,t_{8}  \}  )$ is a triadic concept from the context depicted in Table~\ref{tab:fc3d}. It corresponds to the maximal bicluster  $(\{ g_1, g_2,g_3  \}  , \{ m_1, m_2, m_3  \})$ with $\theta = 1$. $\theta = 1$ since $\{ t_3, t_4, t_5, t_6, t_7 ,t_{8}  \} $ is maximal (it is a modus), it corresponds to interval $[1,2]$ and naturally $2 - 1 = 1$ is the length of this interval.
\end{example} 

\vspace{-0.8cm}
\begin{table}
  \def\X{$\times$}
  \def\O{}
  \centering
  \renewcommand{\arraystretch}{1.1}
  \def\rr#1{\rotatebox{90}{#1}}

\scalebox{0.95}{
\begin{scriptsize}
\begin{tabular}{|r|c|c|c|c|c|c|c|c|c|c|c|c|c|c|c|c|c|c|c|c|c|c|c|c|c|}
\cline{2-26}\multicolumn{1}{r|}{} & \multicolumn{5}{c|}{$t_1 = [0,0]$}                & \multicolumn{5}{c|}{$t_2  = [0,1]$}                & \multicolumn{5}{c|}{$t_3  = [0,2]$}                & \multicolumn{5}{c|}{$t_4 = [0,6]$}                & \multicolumn{5}{c|}{$t_5 = [0,7]$}                 \\
\cline{2-26}\multicolumn{1}{r|}{} & \multicolumn{1}{c}{$m_1$} & \multicolumn{1}{c}{$m_2$} & \multicolumn{1}{c}{$m_3$} & \multicolumn{1}{c}{$m_4$} & $m_5$     & \multicolumn{1}{c}{$m_1$} & \multicolumn{1}{c}{$m_2$} & \multicolumn{1}{c}{$m_3$} & \multicolumn{1}{c}{$m_4$} & $m_5$     & \multicolumn{1}{c}{$m_1$} & \multicolumn{1}{c}{$m_2$} & \multicolumn{1}{c}{$m_3$} & \multicolumn{1}{c}{$m_4$} & $m_5$     & \multicolumn{1}{c}{$m_1$} & \multicolumn{1}{c}{$m_2$} & \multicolumn{1}{c}{$m_3$} & \multicolumn{1}{c}{$m_4$} & $m_5$     & \multicolumn{1}{c}{$m_1$} & \multicolumn{1}{c}{$m_2$} & \multicolumn{1}{c}{$m_3$} & \multicolumn{1}{c}{$m_4$} & $m_5$  \\
\hline
$g_1$    & \multicolumn{1}{c}{\O} & \multicolumn{1}{c}{\O} & \multicolumn{1}{c}{\O} & \multicolumn{1}{c}{\O}  & \O     & \multicolumn{1}{c}{\X} & \multicolumn{1}{c}{\O} & \multicolumn{1}{c}{\O} & \multicolumn{1}{c}{\X}  & \O     & \multicolumn{1}{c}{\X} & \multicolumn{1}{c}{\X} & \multicolumn{1}{c}{\X} & \multicolumn{1}{c}{\X}  & \O     & \multicolumn{1}{c}{\X} & \multicolumn{1}{c}{\X} & \multicolumn{1}{c}{\X} & \multicolumn{1}{c}{\X}  & \X     & \multicolumn{1}{c}{\X} & \multicolumn{1}{c}{\X} & \multicolumn{1}{c}{\X} & \multicolumn{1}{c}{\X}  & \X  \\
$g_2$    & \multicolumn{1}{c}{\O} & \multicolumn{1}{c}{\O} & \multicolumn{1}{c}{\O} & \multicolumn{1}{c}{\X}  & \O     & \multicolumn{1}{c}{\O} & \multicolumn{1}{c}{\X} & \multicolumn{1}{c}{\X} & \multicolumn{1}{c}{\X}  & \O     & \multicolumn{1}{c}{\X} & \multicolumn{1}{c}{\X} & \multicolumn{1}{c}{\X} & \multicolumn{1}{c}{\X}  & \O     & \multicolumn{1}{c}{\X} & \multicolumn{1}{c}{\X} & \multicolumn{1}{c}{\X} & \multicolumn{1}{c}{\X}  & \X     & \multicolumn{1}{c}{\X} & \multicolumn{1}{c}{\X} & \multicolumn{1}{c}{\X} & \multicolumn{1}{c}{\X}  & \X \\
$g_3$    & \multicolumn{1}{c}{\O} & \multicolumn{1}{c}{\O} & \multicolumn{1}{c}{\O} & \multicolumn{1}{c}{\O}  & \O     & \multicolumn{1}{c}{\O} & \multicolumn{1}{c}{\O} & \multicolumn{1}{c}{\X} & \multicolumn{1}{c}{\O}  & \O     & \multicolumn{1}{c}{\X} & \multicolumn{1}{c}{\X} & \multicolumn{1}{c}{\X} & \multicolumn{1}{c}{\O}  & \O     & \multicolumn{1}{c}{\X} & \multicolumn{1}{c}{\X} & \multicolumn{1}{c}{\X} & \multicolumn{1}{c}{\O}  & \X     & \multicolumn{1}{c}{\X} & \multicolumn{1}{c}{\X} & \multicolumn{1}{c}{\X} & \multicolumn{1}{c}{\X}  & \X \\
$g_4$    & \multicolumn{1}{c}{\O} & \multicolumn{1}{c}{\O} & \multicolumn{1}{c}{\O} & \multicolumn{1}{c}{\O}  & \O     & \multicolumn{1}{c}{\O} & \multicolumn{1}{c}{\O} & \multicolumn{1}{c}{\O} & \multicolumn{1}{c}{\O}  & \O     & \multicolumn{1}{c}{\O} & \multicolumn{1}{c}{\O} & \multicolumn{1}{c}{\X} & \multicolumn{1}{c}{\O}  & \O     & \multicolumn{1}{c}{\O} & \multicolumn{1}{c}{\O} & \multicolumn{1}{c}{\X} & \multicolumn{1}{c}{\X}  & \O     & \multicolumn{1}{c}{\O} & \multicolumn{1}{c}{\O} & \multicolumn{1}{c}{\X} & \multicolumn{1}{c}{\X}  & \X  \\
\hline
\multicolumn{1}{r}{} & \multicolumn{1}{c}{} & \multicolumn{1}{c}{} & \multicolumn{1}{c}{} & \multicolumn{1}{c}{} & \multicolumn{1}{c}{} & \multicolumn{1}{c}{} & \multicolumn{1}{c}{} & \multicolumn{1}{c}{} & \multicolumn{1}{c}{} & \multicolumn{1}{c}{} & \multicolumn{1}{c}{} & \multicolumn{1}{c}{} & \multicolumn{1}{c}{} & \multicolumn{1}{c}{} & \multicolumn{1}{c}{} & \multicolumn{1}{c}{} & \multicolumn{1}{c}{} & \multicolumn{1}{c}{} & \multicolumn{1}{c}{} & \multicolumn{1}{c}{} & \multicolumn{1}{c}{} & \multicolumn{1}{c}{} & \multicolumn{1}{c}{} & \multicolumn{1}{c}{} & \multicolumn{1}{c}{}  \\
\cline{2-26}\multicolumn{1}{r|}{} & \multicolumn{5}{c|}{$t_6 = [0,8]$}                & \multicolumn{5}{c|}{$t_7 = [0,9]$}                & \multicolumn{5}{c|}{$t_8 = [1,9]$}                & \multicolumn{5}{c|}{$t_9 = [2,9]$}                & \multicolumn{5}{c|}{$t_{10} = [6,9]$}                 \\
\cline{2-26}\multicolumn{1}{r|}{} & \multicolumn{1}{c}{$m_1$} & \multicolumn{1}{c}{$m_2$} & \multicolumn{1}{c}{$m_3$} & \multicolumn{1}{c}{$m_4$} & $m_5$     & \multicolumn{1}{c}{$m_1$} & \multicolumn{1}{c}{$m_2$} & \multicolumn{1}{c}{$m_3$} & \multicolumn{1}{c}{$m_4$} & $m_5$     & \multicolumn{1}{c}{$m_1$} & \multicolumn{1}{c}{$m_2$} & \multicolumn{1}{c}{$m_3$} & \multicolumn{1}{c}{$m_4$} & $m_5$     & \multicolumn{1}{c}{$m_1$} & \multicolumn{1}{c}{$m_2$} & \multicolumn{1}{c}{$m_3$} & \multicolumn{1}{c}{$m_4$} & $m_5$     & \multicolumn{1}{c}{$m_1$} & \multicolumn{1}{c}{$m_2$} & \multicolumn{1}{c}{$m_3$} & \multicolumn{1}{c}{$m_4$} & $m_5$  \\
\hline
$g_1$    & \multicolumn{1}{c}{\X} & \multicolumn{1}{c}{\X} & \multicolumn{1}{c}{\X} & \multicolumn{1}{c}{\X}  & \X     & \multicolumn{1}{c}{\X} & \multicolumn{1}{c}{\X} & \multicolumn{1}{c}{\X} & \multicolumn{1}{c}{\X} & \X     & \multicolumn{1}{c}{\X} & \multicolumn{1}{c}{\X} & \multicolumn{1}{c}{\X} & \multicolumn{1}{c}{\X} & \X     & \multicolumn{1}{c}{\O} & \multicolumn{1}{c}{\X} & \multicolumn{1}{c}{\X} & \multicolumn{1}{c}{\O} & \X     & \multicolumn{1}{c}{\O} & \multicolumn{1}{c}{\O} & \multicolumn{1}{c}{\O} & \multicolumn{1}{c}{\O} & \X  \\
$g_2$    & \multicolumn{1}{c}{\X} & \multicolumn{1}{c}{\X} & \multicolumn{1}{c}{\X} & \multicolumn{1}{c}{\X}  & \X     & \multicolumn{1}{c}{\X} & \multicolumn{1}{c}{\X} & \multicolumn{1}{c}{\X} & \multicolumn{1}{c}{\X} & \X     & \multicolumn{1}{c}{\X} & \multicolumn{1}{c}{\X} & \multicolumn{1}{c}{\X} & \multicolumn{1}{c}{\O} & \X     & \multicolumn{1}{c}{\X} & \multicolumn{1}{c}{\O} & \multicolumn{1}{c}{\O} & \multicolumn{1}{c}{\O} & \X     & \multicolumn{1}{c}{\O} & \multicolumn{1}{c}{\O} & \multicolumn{1}{c}{\O} & \multicolumn{1}{c}{\O} & \X \\
$g_3$    & \multicolumn{1}{c}{\X} & \multicolumn{1}{c}{\X} & \multicolumn{1}{c}{\X} & \multicolumn{1}{c}{\X}  & \X     & \multicolumn{1}{c}{\X} & \multicolumn{1}{c}{\X} & \multicolumn{1}{c}{\X} & \multicolumn{1}{c}{\X} & \X     & \multicolumn{1}{c}{\X} & \multicolumn{1}{c}{\X} & \multicolumn{1}{c}{\X} & \multicolumn{1}{c}{\X} & \X     & \multicolumn{1}{c}{\X} & \multicolumn{1}{c}{\X} & \multicolumn{1}{c}{\X} & \multicolumn{1}{c}{\X} & \X     & \multicolumn{1}{c}{\O} & \multicolumn{1}{c}{\O} & \multicolumn{1}{c}{\O} & \multicolumn{1}{c}{\X} & \X \\
$g_4$    & \multicolumn{1}{c}{\X} & \multicolumn{1}{c}{\O} & \multicolumn{1}{c}{\X} & \multicolumn{1}{c}{\X}  & \X     & \multicolumn{1}{c}{\X} & \multicolumn{1}{c}{\X} & \multicolumn{1}{c}{\X} & \multicolumn{1}{c}{\X}  & \X     & \multicolumn{1}{c}{\X} & \multicolumn{1}{c}{\X} & \multicolumn{1}{c}{\X} & \multicolumn{1}{c}{\X} & \X     & \multicolumn{1}{c}{\X} & \multicolumn{1}{c}{\X} & \multicolumn{1}{c}{\X} & \multicolumn{1}{c}{\X} & \X     & \multicolumn{1}{c}{\X} & \multicolumn{1}{c}{\X} & \multicolumn{1}{c}{\O} & \multicolumn{1}{c}{\X} & \X  \\
\hline
\multicolumn{1}{r}{}  & \multicolumn{1}{c}{} & \multicolumn{1}{c}{} & \multicolumn{1}{c}{} & \multicolumn{1}{c}{} & \multicolumn{1}{c}{} & \multicolumn{1}{c}{} & \multicolumn{1}{c}{} & \multicolumn{1}{c}{} & \multicolumn{1}{c}{} & \multicolumn{1}{c}{} & \multicolumn{1}{c}{} & \multicolumn{1}{c}{} & \multicolumn{1}{c}{} & \multicolumn{1}{c}{} & \multicolumn{1}{c}{} & \multicolumn{1}{c}{} & \multicolumn{1}{c}{} & \multicolumn{1}{c}{} & \multicolumn{1}{c}{} & \multicolumn{1}{c}{}  \\
\multicolumn{1}{r}{} & \multicolumn{1}{c}{} & \multicolumn{1}{c}{} & \multicolumn{1}{c}{} & \multicolumn{1}{c}{} & \multicolumn{1}{c}{} & \multicolumn{1}{c}{} & \multicolumn{1}{c}{} & \multicolumn{1}{c}{} & \multicolumn{1}{c}{} & \multicolumn{1}{c}{} & \multicolumn{1}{c}{} & \multicolumn{1}{c}{} & \multicolumn{1}{c}{} & \multicolumn{1}{c}{} & \multicolumn{1}{c}{} & \multicolumn{1}{c}{} & \multicolumn{1}{c}{} & \multicolumn{1}{c}{} & \multicolumn{1}{c}{} & \multicolumn{1}{c}{} & \multicolumn{1}{c}{} & \multicolumn{1}{c}{} & \multicolumn{1}{c}{} & \multicolumn{1}{c}{} & \multicolumn{1}{c}{}  \\
\cline{2-16}\multicolumn{1}{r|}{}      & \multicolumn{5}{c|}{$t_{11}= [7,9]$}               & \multicolumn{5}{c|}{$t_{12}= [8,9]$}               & \multicolumn{5}{c|}{$t_{13} = [9,9]$}               & \multicolumn{1}{c}{} & \multicolumn{1}{c}{} & \multicolumn{1}{c}{} & \multicolumn{1}{c}{} & \multicolumn{1}{c}{}  \\
\cline{2-16}\multicolumn{1}{r|}{}  & \multicolumn{1}{c}{$m_1$} & \multicolumn{1}{c}{$m_2$} & \multicolumn{1}{c}{$m_3$} & \multicolumn{1}{c}{$m_4$} & $m_5$     & \multicolumn{1}{c}{$m_1$} & \multicolumn{1}{c}{$m_2$} & \multicolumn{1}{c}{$m_3$} & \multicolumn{1}{c}{$m_4$} & $m_5$     & \multicolumn{1}{c}{$m_1$} & \multicolumn{1}{c}{$m_2$} & \multicolumn{1}{c}{$m_3$} & \multicolumn{1}{c}{$m_4$} & $m_5$     & \multicolumn{1}{c}{} & \multicolumn{1}{c}{} & \multicolumn{1}{c}{} & \multicolumn{1}{c}{} & \multicolumn{1}{c}{}  \\
\cline{1-16}$g_1$     & \multicolumn{1}{c}{\O} & \multicolumn{1}{c}{\O} & \multicolumn{1}{c}{\O} & \multicolumn{1}{c}{\O}  & \O     & \multicolumn{1}{c}{\O} & \multicolumn{1}{c}{\O} & \multicolumn{1}{c}{\O} & \multicolumn{1}{c}{\O}  & \O     & \multicolumn{1}{c}{\O} & \multicolumn{1}{c}{\O} & \multicolumn{1}{c}{\O} & \multicolumn{1}{c}{\O}  & \O     & \multicolumn{1}{c}{} & \multicolumn{1}{c}{} & \multicolumn{1}{c}{} & \multicolumn{1}{c}{} & \multicolumn{1}{c}{}  \\
$g_2$    & \multicolumn{1}{c}{\O} & \multicolumn{1}{c}{\O} & \multicolumn{1}{c}{\O} & \multicolumn{1}{c}{\O}  & \O     & \multicolumn{1}{c}{\O} & \multicolumn{1}{c}{\O} & \multicolumn{1}{c}{\O} & \multicolumn{1}{c}{\O}  & \O     & \multicolumn{1}{c}{\O} & \multicolumn{1}{c}{\O} & \multicolumn{1}{c}{\O} & \multicolumn{1}{c}{\O}  & \O     & \multicolumn{1}{c}{} & \multicolumn{1}{c}{} & \multicolumn{1}{c}{} & \multicolumn{1}{c}{} & \multicolumn{1}{c}{} \\
$g_3$    & \multicolumn{1}{c}{\O} & \multicolumn{1}{c}{\O} & \multicolumn{1}{c}{\O} & \multicolumn{1}{c}{\X} &       & \multicolumn{1}{c}{\O} & \multicolumn{1}{c}{\O} & \multicolumn{1}{c}{\O} & \multicolumn{1}{c}{\O}  & \O     & \multicolumn{1}{c}{\O} & \multicolumn{1}{c}{\O} & \multicolumn{1}{c}{\O} & \multicolumn{1}{c}{\O}  & \O     & \multicolumn{1}{c}{} & \multicolumn{1}{c}{} & \multicolumn{1}{c}{} & \multicolumn{1}{c}{} & \multicolumn{1}{c}{} \\
$g_4$  & \multicolumn{1}{c}{\X} & \multicolumn{1}{c}{\X} & \multicolumn{1}{c}{\O} & \multicolumn{1}{c}{\O} & \X     & \multicolumn{1}{c}{\X} & \multicolumn{1}{c}{\X} & \multicolumn{1}{c}{\O} & \multicolumn{1}{c}{\O}  & \O     & \multicolumn{1}{c}{\O} & \multicolumn{1}{c}{\X} & \multicolumn{1}{c}{\O} & \multicolumn{1}{c}{\O}  & \O     & \multicolumn{1}{c}{} & \multicolumn{1}{c}{} & \multicolumn{1}{c}{} & \multicolumn{1}{c}{} & \multicolumn{1}{c}{}  \\
\cline{1-16}\end{tabular}%
\end{scriptsize}
}
\vspace{0.3cm}
   \caption{Triadic scaled context for Table~\ref{tab:data} with interordinal scaling.}
   \label{tab:fc3d}
\end{table}
\vspace{-0.8cm}

Hence we showed that extracting biclusters of similar values for any $\theta$ in a numerical dataset can be achieved by (i) scaling the attribute value dimension and (ii) extracting the triadic concepts in the resulting derived triadic context.

Interestingly, triadic concepts $(A,B,U)$ with the largest sets $A,B$ or $C$ represent large biclusters of close values. Indeed, the larger $|A|$ and $|B|$ the larger the data covering of the corresponding bicluster. Furthermore, the larger $|U|$, the more similar values for bicluster $(A,B)$. Indeed, by the properties of interordinal scaling, the more intervals in $U$, the smaller their interval intersection. Mining so called top-$k$ frequent triadic concepts can accordingly  be achieved with the existing algorithm {\sc Data-Peeler}~\cite{CerfBRB09}.

On another hand, extracting maximal biclusters for all $\theta$ may be neither efficient nor effective with large numerical data: their number tends to be very large and all biclusters are not relevant for a given analysis. Furthermore, both size and density of contexts derived with interordinal scaling are known to be problematic w.r.t algorithmic scalability, see e.g.~\cite{KaytoueKND11}. In existing methods of the literature, $\theta$ is set \textit{a priori}. We show now how to handle this case with slight modifications, our second main result.

\section{Extracting biclusters of similar values for a given $\theta$}
\label{method2}

In this section we consider the problem of extracting maximal biclusters of similar values in TCA for a given $\theta$ only. It comes with slight modifications of the methodology presented in last section. Intuitively, consider the previous scaling applied on a numerical dataset $(G,M,W,I)$. It scales $W$ into dimension $T$ and subsets of $T$ characterize all intervals of values over $W$. To get maximal biclusters for a given $\theta$ only, we should not consider all possible intervals in $W$, but rather all intervals (i) having a range size that is less or equal than $\theta$ to avoid biclusters with non similar values, and (ii) having a range size the closest as possible to $\theta$ to avoid non-maximal biclusters. For example, if we set $\theta =2$, it is probably not interesting to consider interval $[0,8]$ in the scale dimension since $8-0 > \theta$. Similarly, considering the interval $[6,6]$ may not be interesting as well, since a bicluster with all its values equal to $6$ may not be maximal. As introduced in~\cite{CIKM10}, those maximal intervals of similar values used for the scale are called blocks of tolerance over the set of numbers $W$ with respect to the tolerance relation $\simeq_\theta$.

Therefore we firstly recall basics on tolerance relations over a set of numbers. It allows us to define a simpler scaling procedure. The resulting triadic context is then mined with a new TCA algorithm called {\sc TriMax} to extract maximal biclusters of similar values for a given $\theta$.

Blocks of tolerance over $W$ are defined as maximal sets of pairwise similar values from $W$:
\begin{definition}[Tolerance blocks from a set of numbers]
The similarity relation $\simeq_\theta$ is called a tolerance relation, i.e. reflexive, symmetric but not transitive. Given a set $W$ of values, a subset $V \subseteq W$, and a tolerance relation $\simeq_\theta$ over $W$, $V$ is a \emph{block of tolerance} if:\\
$~~~~~~~~~~$(i)  $\forall w_1,w_2 \in V, ~w_1 \simeq_\theta w_2$ (pairwise similarity) \\
$~~~~~~~~~~$(ii) $\forall w_1 \not \in V, \exists  w_2 \in V,~w_1 \not \simeq_\theta w_2$
(maximality).
\end{definition}

From Table~\ref{tab:data} we have $W = \{0,1,2,6,7,8,9\}$. With $\theta = 2$, one has $0 \simeq_2 2$ but $2 \not \simeq_2 6$. Accordingly, one obtains $3$ blocks of tolerance, namely the sets $\{0,1,2\}$, $\{6,7,8\}$ and $\{7,8,9\}$. These three sets can be renamed as the convex hull of their elements on $\mathbb{N}$: respectively, $[0,2]$, $[6,8]$ and $[7,9]$: any number lying between the minimal and the maximal elements (w.r.t. natural number ordering) of a block of tolerance is naturally similar to any other element of the block.

To derive a triadic context from a numerical dataset, we simply use tolerance blocks over $W$ to define the scale dimension.

\begin{definition}[{\sc Trimax} scale relation] 
The scale relation is a binary relation $J \subseteq W \times C$, where $C$ is the set of blocks of tolerance over $W$ renamed as their convex hulls. Then, $(w,c) \in J$ iff $w \in  c$.
\end{definition}

\begin{example}
From Table~\ref{tab:data} we have: $C = \{[0,1], [1,2], [6,7], [7,8], [8,9]\}$ with $\theta = 1$, and $C = \{ [0,2], [6,8], [7,9] \}$ with $\theta = 2$.
\end{example}

Then, we can apply the same context derivation as in previous section: scaling is still based on intervals, but this time it uses tolerance blocks.
\begin{definition}[{\sc TriMax} triadic scaled context] 
Let $Y \subseteq G \times M \times C$ be a ternary relation. Then $(g,m, c) \in Y$ iff  $(m(g),c) \in J$, or simply $m(g) \in c$, where $J$ is the scale relation. $(G,M,C,Y)$ is called the {\sc TriMax} triadic scaled context.
\end{definition}

\begin{example}
Table \ref{tab:fc3dTriMax} is the {\sc Trimax} triadic scaled concept derived from the numerical dataset lying in Table~\ref{tab:data} with $\theta = 1$.
\end{example}

\vspace{-0.8cm}
\begin{table}
  \def\X{$\times$}
  \def\O{}
  \centering
  \renewcommand{\arraystretch}{1.1}
  \def\rr#1{\rotatebox{90}{#1}}
  \scalebox{0.95}{
\begin{scriptsize}
\begin{tabular}{|r|c|c|c|c|c|c|c|c|c|c|c|c|c|c|c|c|c|c|c|c|c|c|c|c|c|}

\cline{2-26}\multicolumn{1}{r|}{} & \multicolumn{5}{c|}{label 1}                & \multicolumn{5}{c|}{label 2}                & \multicolumn{5}{c|}{label 3}                & \multicolumn{5}{c|}{label 4}                & \multicolumn{5}{c|}{label 5}                 \\
\cline{2-26}\multicolumn{1}{r|}{} & \multicolumn{5}{c|}{$[0,1]$}                & \multicolumn{5}{c|}{$[1,2]$}                & \multicolumn{5}{c|}{$[6,7]$}                & \multicolumn{5}{c|}{$[7,8]$}                & \multicolumn{5}{c|}{$[8,9]$}                 \\
\cline{2-26}\multicolumn{1}{r|}{} & \multicolumn{1}{c}{$m_1$} & \multicolumn{1}{c}{$m_2$} & \multicolumn{1}{c}{$m_3$} & \multicolumn{1}{c}{$m_4$} & $m_5$     & \multicolumn{1}{c}{$m_1$} & \multicolumn{1}{c}{$m_2$} & \multicolumn{1}{c}{$m_3$} & \multicolumn{1}{c}{$m_4$} & $m_5$     & \multicolumn{1}{c}{$m_1$} & \multicolumn{1}{c}{$m_2$} & \multicolumn{1}{c}{$m_3$} & \multicolumn{1}{c}{$m_4$} & $m_5$     & \multicolumn{1}{c}{$m_1$} & \multicolumn{1}{c}{$m_2$} & \multicolumn{1}{c}{$m_3$} & \multicolumn{1}{c}{$m_4$} & $m_5$     & \multicolumn{1}{c}{$m_1$} & \multicolumn{1}{c}{$m_2$} & \multicolumn{1}{c}{$m_3$} & \multicolumn{1}{c}{$m_4$} & $m_5$  \\
\hline

$g_1$     
& \multicolumn{1}{c}{\X} & \multicolumn{1}{c}{\O} & \multicolumn{1}{c}{\O} & \multicolumn{1}{c}{\X}  & \O    

& \multicolumn{1}{c}{\X} & \multicolumn{1}{c}{\X} & \multicolumn{1}{c}{\X} & \multicolumn{1}{c}{\X}  & \O    

 & \multicolumn{1}{c}{\O} & \multicolumn{1}{c}{\O} & \multicolumn{1}{c}{\O} & \multicolumn{1}{c}{\O}  & \X    

 & \multicolumn{1}{c}{\O} & \multicolumn{1}{c}{\O} & \multicolumn{1}{c}{\O} & \multicolumn{1}{c}{\O}  & \O  

 & \multicolumn{1}{c}{\O} & \multicolumn{1}{c}{\O} & \multicolumn{1}{c}{\O} & \multicolumn{1}{c}{\O}  & \O  \\

$g_2$      & \multicolumn{1}{c}{\O} & \multicolumn{1}{c}{\X} & \multicolumn{1}{c}{\X} & \multicolumn{1}{c}{\X}  & \O    
& \multicolumn{1}{c}{\X} & \multicolumn{1}{c}{\X} & \multicolumn{1}{c}{\X} & \multicolumn{1}{c}{\O}  & \O   
& \multicolumn{1}{c}{\O} & \multicolumn{1}{c}{\O} & \multicolumn{1}{c}{\O} & \multicolumn{1}{c}{\O}  & \X    
 & \multicolumn{1}{c}{\O} & \multicolumn{1}{c}{\O} & \multicolumn{1}{c}{\O} & \multicolumn{1}{c}{\O}  & \O    
 & \multicolumn{1}{c}{\O} & \multicolumn{1}{c}{\O} & \multicolumn{1}{c}{\O} & \multicolumn{1}{c}{\O}  & \O \\

$g_3$     & \multicolumn{1}{c}{\O} & \multicolumn{1}{c}{\O} & \multicolumn{1}{c}{\X} & \multicolumn{1}{c}{\O}  & \O   
& \multicolumn{1}{c}{\X} & \multicolumn{1}{c}{\X} & \multicolumn{1}{c}{\X} & \multicolumn{1}{c}{\O}  & \O      
& \multicolumn{1}{c}{\O} & \multicolumn{1}{c}{\O} & \multicolumn{1}{c}{\O} & \multicolumn{1}{c}{\X}  & \X    
 & \multicolumn{1}{c}{\O} & \multicolumn{1}{c}{\O} & \multicolumn{1}{c}{\O} & \multicolumn{1}{c}{\X}  & \O     
& \multicolumn{1}{c}{\O} & \multicolumn{1}{c}{\O} & \multicolumn{1}{c}{\O} & \multicolumn{1}{c}{\O}  & \O \\

$g_4$     & \multicolumn{1}{c}{\O} & \multicolumn{1}{c}{\O} & \multicolumn{1}{c}{\O} & \multicolumn{1}{c}{\O}  & \O     
& \multicolumn{1}{c}{\O} & \multicolumn{1}{c}{\O} & \multicolumn{1}{c}{\X} & \multicolumn{1}{c}{\O}  & \O       
& \multicolumn{1}{c}{\O} & \multicolumn{1}{c}{\O} & \multicolumn{1}{c}{\O} & \multicolumn{1}{c}{\X}  & \X    
 & \multicolumn{1}{c}{\X} & \multicolumn{1}{c}{\O} & \multicolumn{1}{c}{\O} & \multicolumn{1}{c}{\O}  & \X    
 & \multicolumn{1}{c}{\X} & \multicolumn{1}{c}{\X} & \multicolumn{1}{c}{\O} & \multicolumn{1}{c}{\O}  & \O  \\
\hline
\end{tabular}%
\end{scriptsize}
}
\vspace{0.3cm}
   \caption{Triadic scaled context using tolerance blocks over $W$ and $\theta=1$.}
   \label{tab:fc3dTriMax}
\end{table}
\vspace{-1cm}

\begin{definition}[Dyadic context associated with a block of tolerance]
Consider a block of tolerance $c \in C$. The dyadic context associated with this block is given by $(G, M, Z)$ where $z \in Z$ denotes all $(g,m) \in G \times M$ such as $ m(g) \in c$.
\end{definition}
\begin{example}
In Table \ref{tab:fc3dTriMax}, each such dyadic context is labelled by its corresponding block of tolerance.
\end{example}

Now, remark that blocks of tolerance over $W$ are totally ordered: let $[v_1, v_2]$ and $[w_1,w_2]$ be two blocks of tolerance, one has $[v_1, v_2] < [w_1,w_2]$ iff $v_1 < w_1$. Hence, associated dyadic contexts are also totally ordered and we use a corresponding indexing set to label them. In  Table \ref{tab:fc3dTriMax}, contexts for blocks $\langle [0,1], [1,2]$, $[6,7]$, $[7,8], [8,9] \rangle$ are respectively labelled $\langle 1,2,3,4,5 \rangle$.

We now present our second main results: The scaled triadic context supports the extraction of maximal biclusters of similar values for a given $\theta$. In this case however, existing algorithms of TCA cannot be applied directly. For example, in Table~\ref{tab:fc3dTriMax}, the triconcept $(\{g_3\},\{m_4\}, \{3,4\})$ corresponds to a bicluster of similar values  which is not maximal. Hence we present hereafter a new TCA algorithm for this task, called {\sc TriMax}. 

The basic idea of {\sc TriMax} relies on the following facts. Firstly, since each dyadic context corresponds to a block of tolerance, we do not need to compute intersections of contexts, such as classically done in TCA. Hence each dyadic context is processed separately. Secondly, a dyadic concept of a dyadic context necessarily represents a bicluster of similar values, but we cannot be sure it is maximal (see previous example). Hence, we need to check if a concept is still a concept in other dyadic contexts, corresponding to other classes of tolerance. This is made precise with the following proposition.

\begin{proposition}
Let  $(A,B,U)$ be a triadic concept from {\sc Trimax} triadic scaled context $(G,M,C,Y)$, such that $U$ is the outer closure of a singleton $\{c\} \subseteq C$. If $|U|=1$, $(A,B)$ is a maximal bicluster of similar values. Otherwise, $(A,B)$ is a maximal bicluster of similar values iff $\nexists y \in [min (U) ; max (U)]$, $y < c$ s.t. $(A,B) \not = \Psi^{'}_y( \Psi_y((A,B)))$, where $\Psi^{'}_y(.)$ and $\Psi_y(.)$ correspond to inner derivation operators associated with $y^{th}$ dyadic context.
\label{trimax-concepts}
\end{proposition}
\proof{
When $|U|=1$, $(A,B)$ is a dyadic concept only in one dyadic context corresponding to a block of tolerance. By properties of tolerance blocks, $(A,B)$ is a maximal bicluster. If $|U| \not =1$,  $(A,B)$ is a dyadic concept in $|U|$ dyadic contexts. Since the tolerance block set is totally ordered, it directly implies that modus $U$ is an interval $[min (U) ; max (U)]$. Hence, if  $\exists y \in [min (U) ; max (U)]$ s.t. $(A,B) = \Psi^{'}_y( \Psi_y((A,B)))$ this means that $(A,B)$ is not a maximal bicluster of similar values.
}

\medskip

\textbf{Description of the {\sc TriMax} algorithm.}
{\sc TriMax} starts with scaling initial numerical data into several dyadic contexts, each one standing for a block of tolerance over $W$ with given $\theta$.  The set of all dyadic contexts forms accordingly a triadic context. Then, each dyadic context is mined with any FCA algorithm (or closed itemset mining algorithm), and all formal concepts are extracted. For a given concept $(A,B)$, we compute outer derivation  $\Phi^{'}( (A,B))$, i.e. to obtain  the set of dyadic contexts labels in which the current dyadic concept holds. If it results in a singleton, this means that $(A,B)$ is a concept for the current block of tolerance only, i.e. it is a maximal bicluster of similar values, and it has been, or will never be, generated twice. Otherwise, $(A,B)$ is a concept in other contexts, and can be generated accordingly several times (as much as the number of contexts in which it holds). Then, we only consider $(A,B)$ if we are sure it is the last time it is computed. Finally, we need to check if current concept represents a maximal bicluster, i.e. there should not exist a context from the modus where $(A,B)$ is not a dyadic concept.

\begin{algorithm}
\SetKwData{Left}{left}\SetKwData{This}{this}\SetKwData{Up}{up}
\SetKwFunction{Union}{Union}\SetKwFunction{FindCompress}{FindCompress}
\SetKwInOut{Input}{input}\SetKwInOut{Output}{output}
\Input{Numerical dataset  $(G,M,W,I)$, tolerance parameter $\theta$}
\Output{Maximal biclusters of similar values}
\BlankLine

Let $C = \{ [a_i, b_i ] \}$ be the totally ordered set of all blocks over $W$ for given $\theta$. Indices $i$ form an indexing set.

\ForAll{$[a_i,b_i] \in C$}{
Build context $(G,M, Z_i)$ such that $(g,m) \in Z_i \Leftrightarrow m(g) \in [a_i,b_i]$
}

\ForAll {$(G,M, Z_i)$}{
Use any FCA algorithm to extract all its concepts $(A,B)$\\
\ForAll {dyadic concepts $(A,B)$ in the current context $(G,M, Z_i)$ }
{
\If {$ | \Phi^{'}( (A,B)) | = 1$}
{
print $(A,B)$
}
\ElseIf{$max (\Phi^{'}( (A,B)) = i$}
{
	$x \leftarrow  min (\Phi^{'}( (A,B)) $ \\

	\If {$\nexists y \in [x, i[$ s.t. $(A,B) \not = \Psi^{'}_y( \Psi_y((A,B)))$ }	
	{
	print $(A,B)$
	}
}
}
}
\caption{TriMax}\label{trimax}
\end{algorithm}

\begin{proposition}
{\sc TriMax} outputs a (i) complete, (ii) correct and (iii) non redundant collection of all maximal biclusters of similar values for a given numerical dataset and similarity parameter $\theta$.
\end{proposition}
\proof{
(i) and (ii) follow directly from Proposition \ref{trimax-concepts}. Statement (iii) is ensured by the second \textit{if} condition of the algorithm: a dyadic concept (or equivalently bicluster) is considered iff it has been extracted in the last dyadic context in which it holds. 
}

\section{Computer experiments}
\label{experiments}
In this section, we experiment with the algorithm {\sc TriMax} and highlight various aspects of its practical complexity. 

\medskip

\noindent\textbf{Data.} We explore a gene expression dataset of the species \textit{Laccaria bicolor} available at NCBI\footnote{http://www.ncbi.nlm.nih.gov/geo/ as series GSE9784}. More details on this dataset can be found in~\cite{KaytoueKND11}. This gene expression dataset monitors the behaviour of $11, 930$ genes in $12$ biological situations, reflecting various stages of \textit{Laccaria bicolor} biological cycle. Attribute values in $W$ vary between $0$ and $60, 000$.

\medskip

\noindent\textbf{{\sc TriMax} implementation}. {\sc TriMax} is written in C++. It uses the {\sc boost} library 1.42 for data structures and the implementation of {\sc InClose} from its authors\footnote{\url{http://sourceforge.net/projects/inclose/}} for dyadic concepts extraction.
At each iteration of the main loop, i.e. each tolerance block, the current scaled dyadic context is produced: We do not generated the whole triadic context which cannot fit into memory for large databases. It turns out that the modus computation for a given dyadic concept requires to compute scaling ``on the fly'', i.e. when computing the set of dyadic contexts in which a current concept holds. The experiments were carried out on an Intel CPU 2.54 Ghz machine with 8 GB RAM running under Ubuntu 11.04. 

\medskip

\begin{figure} 
\begin{tabular}{cc}
\includegraphics[scale=0.26]{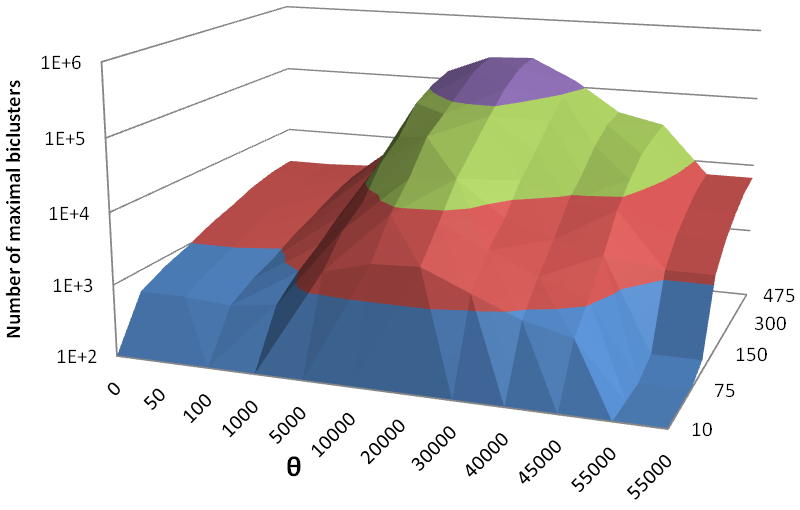}  & \includegraphics[scale=0.26]{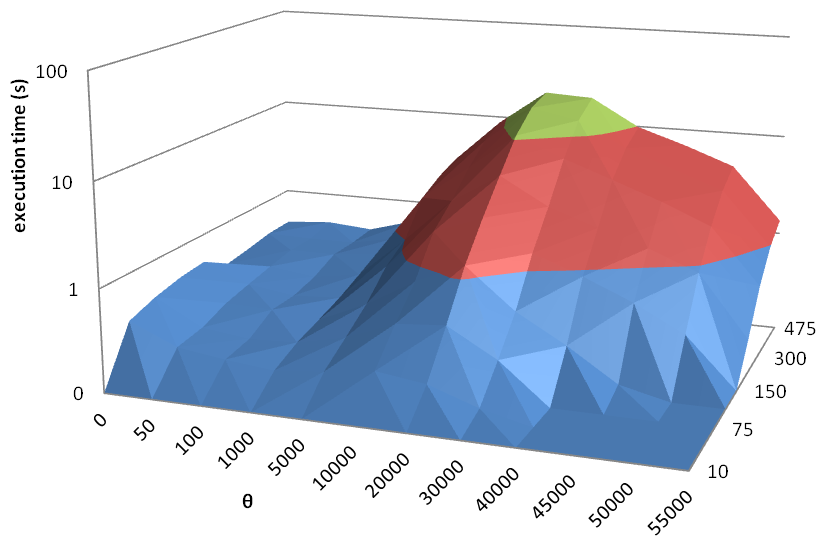}  \\ 
  ~& ~\\
(i) Numbers of patterns (Y-axis)  &  (ii) Execution times in seconds (Y-axis)  \\
  w.r.t. $\theta$ (X-axis) and $|G|$ (Z-axis) & w.r.t. $\theta$ (X-axis) and $|G|$ (Z-axis) \\
  ~& ~\\
\includegraphics[scale=0.26]{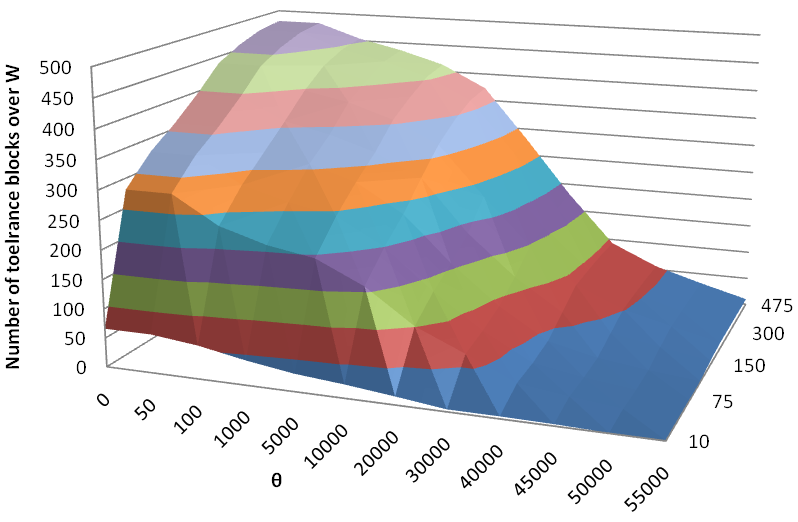} &\includegraphics[scale=0.26]{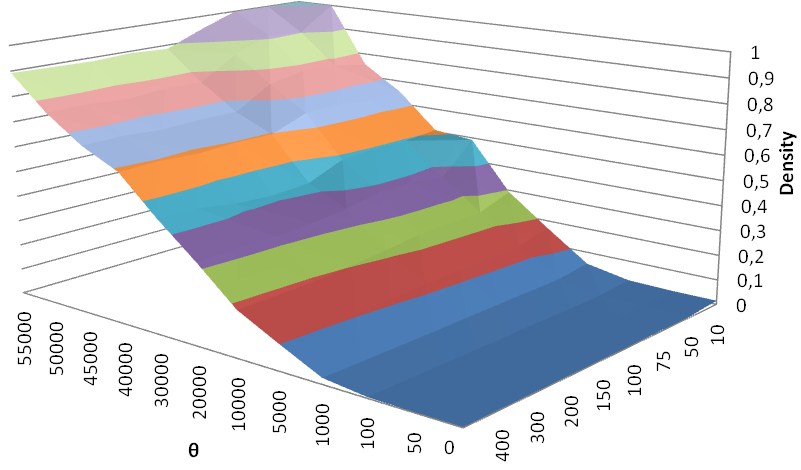}  \\ 
  ~& ~\\
(iii) Numbers of blocks of tolerance (Y-axis) & (iv) Density of triadic contexts (Y-axis) \\
  w.r.t. $\theta$ (X-axis) and $|G|$ (Z-axis) & w.r.t. $\theta$ (X-axis) and $|G|$ (Z-axis) \\
    ~& ~\\
\includegraphics[scale=0.32]{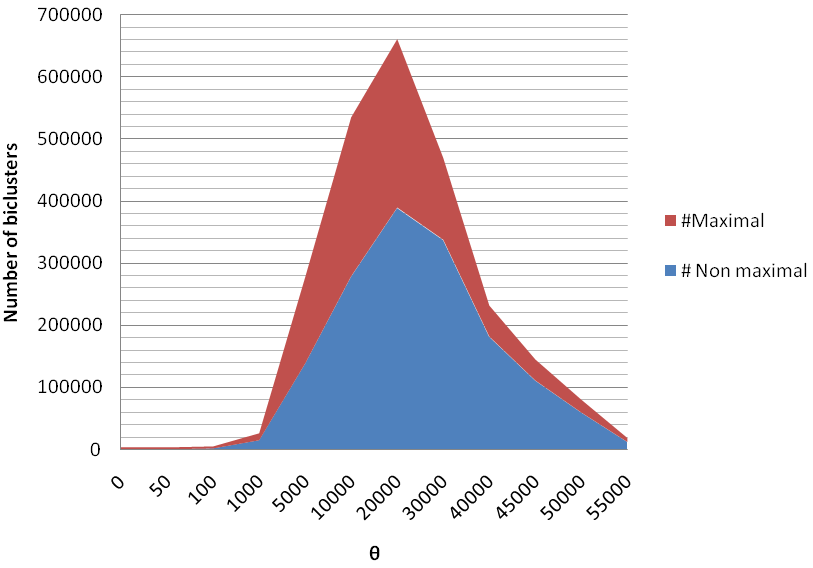} & \includegraphics[scale=0.29]{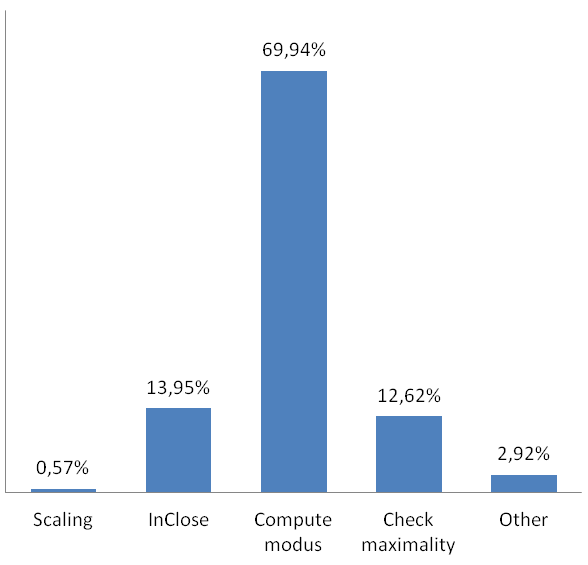}  \\ 
(v) Comparing the number of generated dyadic  & (vi) Repartition of execution time \\
concepts w.r.t. the actual number of maximal  &  w.r.t main steps of {\sc TriMax} \\
biclusters varying $\theta$ with $|G| = 500$  & with $\theta = 33, 000$ and $|G| = 500$ \\
  ~& ~\\
\end{tabular} 
\caption{Monitoring with different settings (i) the number of maximal biclusters, (ii) the execution times of {\sc TriMax}, (iii) the number of tolerance blocks, (iv) the derived triadic context density, (v) the number of non-maximal biclusters generated as dyadic-concepts w.r.t. the number of maximal biclusters, and (vi) repartition of execution time in the {\sc TriMax} algorithm.}
\label{figs}
\end{figure}
\noindent\textbf{Experiment settings.}
The goal of the present experiments is not to give a qualitative evaluation of the present approach (say biological interpretation), but rather a quantitative evaluation. Indeed, the present work aims at showing how an existing type of biclusters can be mined with Triadic Concept Analysis. For a qualitative evaluation, the reader may refer for example to~\cite{BessonRRB06,KaytoueKND11}.

Accordingly, we designed the following experiments to monitor various aspects of the {\sc TriMax} algorithm. For most of the experiments, the dataset used is composed of an increasing number of objects and all attributes. The objects are chosen randomly once and for all so that the different experiment results can be compared. We also vary the parameter $\theta$ in the same way across all experiments. Then, we monitor the following aspects, as presented in Figure~\ref{figs}:
\vspace{-0.2cm}
\begin{enumerate}[i.]
\item Number of maximal biclusters of similar values
\item Execution time (in seconds)
\item Number of tolerance blocks
\item Density of the triadic context, where density is defined as $d(G,M,C,Y) = |Y| / (|G| \times |M| \times |C|)$. This information is important, since contexts with high density are known to be hard to process with FCA algorithms~\cite{KuznetsovO02}, and we use the {\sc InClose} algorithm for dyadic contexts processing.
\item Comparison between the number of non-maximal biclusters produced by {\sc TriMax} (i.e. dyadic concepts that do not corresponds to maximal biclusters) with the number of maximal biclusters.
\item Execution time profiling of the main procedures of {\sc TriMax}. This is achieved with the tool {\sc GNU GProf} and gives us what parts of the algorithm are the most time consuming.
\end{enumerate}
 
\medskip

\noindent\textbf{Experiment results.} Figure~\ref{figs} presents the results of our experiments with different settings. In these settings, we vary the number of objects $|G|$ and the parameter $\theta$. A first observation arises from graph (i): the number of biclusters is the highest when $\theta \simeq 30, 000$. A first explanation is that $30, 000$ is the half of the maximal value of $W$ and almost all multiples of $100$ in $[0;60, 000]$ belongs to $W$. In graph (ii), execution time has the same behaviour as graph (i). These results can be understood by paying attention to the next graphs (iii) and (iv). In (iii) is monitored the number of tolerance blocks. The maximal number is reached when $\theta = 0$, i.e. $|C| = |W|$. When $\theta = max(W)$, we have $|C| = 1$. Now we observe in (iv) that the density follows a reverse behaviour: When $\theta = 0$, the density tends towards $0\%$; when $\theta = max(W)$, then density exactly equal $1\%$. Combining both graph (iii) and (iv), the worst cases happen when both density and tolerance bloc count are high.

Another observation, which explains also the execution times, arises from graph (v). Here are compared the number of maximal biclusters and the number of non-maximal biclusters generated as dyadic concepts. Here again, worst case is reached when $\theta \simeq 30, 000$. Looking at graph (vi), we learn that this is however not the major problem. The mostly consuming procedure of {\sc TriMax} is the computation of the modus of a dyadic concept. The explanation is that we compute modus with ``on the fly scaling''. 

Therefore, the bottleneck of our algorithm reveals itself to be the modus computation. In practical applications however, the analyst is not interested in all biclusters of similar values. Some constraints are generally defined, such as a minimal (resp. maximal) number of objects (resp. attributes) in a bicluster $(A,B)$, or a minimal  area $|A| \times |B|$, etc. Interestingly, most of those constraints can be evaluated on a generated dyadic concept. Therefore, before computing the modus of such concept, we can check such properties and discard the concept if not respecting the constraints. Although not reflected in this paper, we tested how adding minimal (resp. maximal) size constraints on a bicluster affects both number of biclusters and execution times. The results are very interesting: for example with $\theta = 33, 000$, $|G| = 500$, and minimal (resp. maximal) size for $|A|$ set to $10$ (resp. $40$), {\sc TriMax} produces only $5, 332$ maximal biclusters in $2.1$ seconds compared to $104, 226$ maximal biclusters extracted in $16.130$ seconds without any constraint.

Finally, the most interesting aspect of {\sc TriMax} is its direct distributed computation capacity. Indeed, each iteration, i.e. for each block of tolerance, can be achieved independently from the others. Furthermore, the core of {\sc TriMax} consisting in extracting dyadic contexts can also be distributed, see e.g.~\cite{Krajca09}. A deeper investigation remains to be done in this case. Note that although the method description involves $W$ as a set of natural numbers, {\sc TriMax} can directly handle numerical data real numbers, and has been implemented as such.

\medskip

\noindent\textbf{Comparison with existing methods.} Two existing methods in the literature also consider the problem of extracting all maximal biclusters of similar values from a numerical dataset. The first method is called \textit{Numerical Biset Miner} ({\sc NBS-Miner}~\cite{BessonRRB06}). The second method is based on \textit{interval pattern structures} ({\sc IPS}~\cite{KaytoueKN11,IJCAI11}). Limited by space, we do not detail these methods. Both {\sc NBS-Miner} and {\sc IPS} algorithms have been implemented in C++. 
First experiments show that {\sc NBS-Miner} is not scalable compared to {\sc IPS} and {\sc TriMax}. On another hand, it seems that {\sc TriMax} outperforms {\sc IPS}, but a deeper investigation is required. The main problem in {\sc IPS} is to find an efficient algorithm able to compute tolerance blocks over a set of intervals.

\section{Conclusion}
We addressed the problem of biclustering numerical data with Formal Concept Analysis. So called (maximal) biclusters of similar values can be characterized and extracted with Triadic Concept Analysis, which turns out to be a novel mathematical framework for this task. We properly defined a scaling procedure turning original numerical data into triadic contexts from which biclusters can be extracted as triadic concepts with existing algorithms. This approach allows a correct, complete and non-redundant extraction of all maximal biclusters, for any similarity parameter $\theta$ and can be extended to $n$-ary numerical datasets while their computation can be directly distributed. The interpretation of triadic concepts is very rich: both extent and intent allow to characterize a bicluster (i.e. the rectangle), while the modus gives the range of values of the biclusters, and for which $\theta$ is the bicluster maximal. Moreover, the larger the modus, the more similar the values within current bicluster. It follows a perspective of research, aiming at extracting the top-$k$ frequent tri-concepts with {\sc Data-Peeler}~\cite{CerfBRB09}, which can help to handle the problem of top-$k$ biclusters extraction. We also adapted the TCA machinery with algorithm {\sc TriMax} to extract maximal biclusters for a user-defined $\theta$, which is classical in the existing literature. It appears that {\sc TriMax} is a fully customizable algorithm: any concept extraction algorithm can be used inside its core (along with several constraints on produced dyadic concepts), while its distributed computation is direct. Among several other experiments, it remains now to determine which are the best core algorithms for a given $\theta$ parameter, the very last directly influencing derived contexts density.

\medskip

\noindent\textbf{Acknowledgements.}  Authors would like to thank Dmitry Andreevich Morozov for implementing the algorithms NBS-Miner and IPS. Mehdi Kaytoue was partially supported by CNPq,  Fapemig and the Brazilian National
Institute for Science and Technology for the Web (InWeb). Sergei O. Kuznetsov was supported by the project of the Russian Foundation for Basic Research, grant no. 08-07-92497-NTsNIL$\_$a. Juraj Macko acknowledges support by Grant No. \mbox{202/10/0262} of the Czech Science Foundation.

\appendix
\section{Proof of the Proposition 1.}
Before proving this proposition, we need to introduce the following.
For sake of simplicity, we now consider $W$ as the set of all natural numbers from a numerical dataset that are greater or equal than the minimal value and lower or equal than the maximal value, i.e. $ W = \{0,1,2,3,4,5,6,7,8,9\}$ with the example of Table~\ref{tab:data}.

\begin{definition}[Scale value and scale relation] 
We call \textit{scale value} $s=q-r$ where $r = min(W)$ and $q = max(W)$. The scale relation is a binary relation $J \subseteq W \times T$, where $T=\{t_1, \dots, t_{2s+1} \}$  $r \leq w \leq q$ and  $\tu{ w,t_i } \in J$ iff $i \in  [w-r+1,w-r+1+s]$.
\end{definition}

Note that  $J$ is equivalent to interordinal scale of $W$ previously given, but this notations are used for the proof.

\begin{definition}[$E_{\theta w}$ - cluster base] We introduce 
$E_{\theta w} \subseteq T$ defined as $E_{\theta w}=[t_{w+\theta-r+1};t_{w-r+1+s}]$ for given $\theta$ and $w \in W$.
\end{definition} 

\begin{example}[$E_{\theta w}$ - cluster base] 
$ E_{1 2}=[t_{2+1-0+1};t_{2-0+1+9}]=[t_{4};t_{12}]$.
\end{example} 

\begin{proposition}
$(w_b=m(g)) \simeq_{\theta}(n(h)=w_c)$  iff $(\tu{g,m} \in Y^{12}_{E_
{\theta b}}$ \mbox{ and } $\tu{h,n} \in Y^{12}_{E_{\theta b}})$.
\end{proposition} 

\begin{proof}
Let $E_b,E_c \subseteq T$ and $w_c \geq w_b$. According to the definition $(g,m)\in Y^{12}_{E_{\theta b}}$ 
iff $ m,g,t$ are related by $Y$  for all $ t \in E_{\theta b} $. Using scaling and definition we have $[t_{w_b-r+1};t_{w_b-r+1+s}]=E_b\supseteq E_{\theta b}=[t_{w_b+\theta-r+1};t_{w_b-r+1+s}]$ which is straightforward. We just need to show that $(h,n)\in Y^{12}_{E_{\theta b}}$ holds as well. With scaling definition and previous definition we get $[t_{w_c-r+1};t_{w_c-r+1+s}]=E_c\supseteq E_{\theta b}=[t_{w_b+\theta-r+1};t_{w_b-r+1+s}]$ holding iff  $w_c-w_b \leq \theta$, which is equal to the definition of $\simeq_{\theta}$.
\end{proof}
Moreover we can easily see as a corollary  that  $w_c-w_b \leq \theta$ holds iff $E_b \cap E_c \supseteq E_{\theta b}$ and  $w_c-w_b = \theta$ holds iff $E_b \cap E_c = E_{\theta b}$.
Now we can prove the Proposition 1 from the main text.

\setcounter{proposition}{0}

\begin{proposition}
Tuple $ \tu{A_1,A_2,U} $, where $A_1 \subseteq G$, $A_2 \subseteq M$ and $U \subseteq T$ is triadic concept iff $(A_1,A_2)$ is a maximal bicluster of similar values for some $\theta \geq 0$. Furthermore the value of $\theta$ is defined as $\theta=s-|U|+1$. 
\end{proposition} 

\begin{proof}
Let $U=E_{\theta b}$ and consider dyadic context $ Y^{12}_{U}= Y^{12}_{E_{\theta b}}$ for some $w_b$. Using dyadic closure operator $ \Psi^{'}(\Psi((A_1))$ we get  $(A_1,A_2)$. From definition of triconcept we know that $A_1 \subseteq B_1$ implies $A_1=B_1$ (the same for $A_2$). From definition of maximal bicluster of similar values we know that $\tu{A_1,A_2}$ is maximal when it does not exists $\tu {B_1,B_2}$ s.t. $B_1 \supseteq A_1$ (the same applies for $A_2$). It is obvious that both sets are maximal from definition and when we have the same dyadic context  $ Y^{12}_{U}= Y^{12}_{E_{\theta b}}$. Now we need to look at dyadic context  $ Y^{12}_{U}= Y^{12}_{E_{\theta b}}$. In $|U|=|E_{\theta b}|=|[t_{w_b+\theta-r+1};t_{w_b-r+1+s}]|$ we can easily see  that $|U|=s-\theta+1$, which gives $\theta=s-|U|+1$.  

Finally, $U$ is maximal (as being modus of a triconcept) and $E_{\theta b}$ is maximal as well because  $w_c-w_b \leq \theta$ holds iff $E_b \cap E_c \supseteq E_{\theta b}$. All facts mentioned in this proof leads to equality of the triconcept and maximal bicluster of similar values. 
\end{proof}
 
\bibliographystyle{splncs03}
\bibliography{CLA-11}{}
\end{document}